\documentclass[sigconf, nonacm, 10pt]{acmart}
\pdfoutput=1
\usepackage{graphicx}
\usepackage{subcaption}

\usepackage[vlined,ruled,linesnumbered]{algorithm2e}
\usepackage{amsmath}

\newcommand{\eg}{\textit{e.g., }} 
\newcommand{\ie}{\textit{i.e., }}   
\newcommand{\etc}{\textit{etc.}} 
\newcommand{\whp}{\textit{w.h.p. }}

\newcommand{\fission}{\textsf{Fission }} 
\newcommand{\fissionb}{\textsf{Fission}} 
\newcommand{\fissions}{\textsf{Fission's }}

\newcommand{\fitb}{\textsf{FIT}} 

\newcommand{\ppap}{\textsf{PPAP }} 
\newcommand{\ppapb}{\textsf{PPAP}}

\DeclareMathOperator*{\E}{E}

\DeclareMathOperator*{\Var}{Var}

\DeclareMathOperator*{\argmin}{argmin}
\DeclareMathOperator*{\loglog}{loglog}

\renewcommand\footnotetextcopyrightpermission[1]{}
\begin{document}
\title{Fission: A Provably Fast, Scalable, and Secure Permissionless Blockchain}
\subtitle{Version 1.0}
\author{Ke Liang}
\affiliation{%
	\institution{The Fission Project}
}
\email{ke@fission.org}

\begin{abstract}
We present \fissionb, a new permissionless blockchain that achieves scalability in both terms of system throughput and transaction confirmation time, while at the same time, retaining blockchain's core values of equality and decentralization.

\fission overcomes the system throughput bottleneck by employing a novel Eager-Lazy pipeling model that achieves very high system throughputs via block pipelining, an adaptive partitioning mechanism that auto-scales to transaction volumes, and a provably secure energy-efficient consensus protocol to ensure security and robustness. 

\fission applies a hybrid network which consists of a relay network, and a peer-to-peer network. The goal of the relay network is to minimize the transaction confirmation time by minimizing the information propagation latency. To optimize the performance on the relay network in the presence of churn, dynamic network topologies, and network heterogeneity, we propose an ultra-fast game-theoretic relay selection algorithm that achieves near-optimal performance in a fully distributed manner. \fissions peer-to-peer network complements the relay network and provides a very high data availability via enabling users to contribute their storage and bandwidth for information dissemination (with incentive). We propose a distributed online data retrieval strategy that optimally offloads the relay network without degrading the system performance.

By re-innovating all the core elements of the blockchain technology - computation, networking, and storage - in a holistic manner, \fission aims to achieve the best balance among scalability, security and decentralization.
\end{abstract}

\maketitle

\section{Introduction}
Due to the phenomenal success of cryptocurrencies \cite{nakamoto2008bitcoin, wood2014ethereum} in the last few years, blockchain technology has gained massive interests, and recently emerged with a promise to streamline interactions in a wide range of settings. A blockchain, also called a distributed ledger, is a decentralized and incorruptible ledger used to record transactions across many nodes\footnote{We use the terms node, peer, and user interchangeably.} which do not need to fully trust each other. By applying certain consensus mechanisms, all the nodes in a blockchain network agree on an ordered set of blocks, each of which could contain multiple transactions. It should be noted that we only discuss permissionless or public blockchains in this paper, where anyone can join and participate in the process of transaction verification and block proposing. 

The key feature of blockchains is \emph{decentralization}. That is, instead of relying on central trusted authorities or infrastructures, blockchains are built on top of a global peer-to-peer (P2P) network, where messages (\eg transactions, blocks) are disseminated to the whole network in a gossip-like manner. Moreover, everyone can verify all the transactions and propose new blocks which are supposed to be appended to the blockchain via consensus mechanisms. Note that these transactions are not just financial transactions (\eg cryptocurrencies, tokens), but virtually everything of value. By applying new applications, like smart contracts \cite{ethereum2016}, blockchains have been leveraged to many sectors, \eg financial services, Internet-of-Things, insurance, shifting the landscapes of the industries that worth trillions of dollars. 

Despite the fact that the blockchain technology brings significant opportunities and disruptive potential in many industries, \emph{scalability} has been a key issue severely that limits the adoptions of the blockchain technology, causing poor user experience, congested network, and skyrocketing transaction fees. The scalability of blockchains is fundamentally hindered by the challenges of distributed system design (\eg consensus protocols), and limitations of the underlying P2P networks. More specifically, the former results from the core problem of all blockchains, \ie double-spending. It is difficult to prevent double-spending by achieving consensus on an ordered list of transactions in a public setting where anyone, including malicious users, can participate, anytime, anywhere. The latter is due to the fact that P2P network protocols (\eg Kademlia DHT \cite{maymounkov2002kademlia,anderson2016new}) that most existing blockchains employ are not designed for blockchains, where messages (\ie transactions and blocks) need to be disseminated constantly in a many-to-many manner. Therefore, there are two important challenges that a scalable blockchain should address: 1) high throughput (measured by the transaction rate, \ie TPS) and 2) fast confirmation times\footnote{A transaction is confirmed if it is included in a confirmed block.}. 

To address the first challenge, a considerable amount of Proof-of-Work (PoW) based consensus protocols~\cite{Sompolinsky2016SPECTREAF,sompolinsky2018phantom,eyal2016bitcoin,li2018scaling} and Proof-of-Stake (PoS) based consensus protocols~\cite{bentov2016snow,david2018ouroboros,Gilad:2017:ASB:3132747.3132757,buterin2017casper} have been proposed to increase the system throughput without compromising the security guarantees and decentralization. However, all of the above new PoS-based consensus protocols can only achieve sub-optimal performance in terms of system throughput and confirmation times due to 1) the constrained resources (computation, bandwidth, memory, \etc) of nodes that are selected for transaction verification and block proposing, and 2) the limitations of underlying P2P networks which results in high information propagation latency due to dynamic network topologies and huge blocks with a large number of transactions. To further improve the system throughput, sharding technology \cite{kokoris2018omniledger,hsiao-wei2017,zamani2018rapidchain} has been proposed to split up the task of consensus among multiple, smaller concurrently operating sets of nodes, thus reducing per-node processing and storage requirements. However, creating a secure \textit{sharding} solution that is capable of making cross-shard (or inter-shard) transactions (especially atomic synchronous transactions) is non-trivial, and existing sharding-based consensus protocols either make security/performance trade-offs or rely on strong assumptions (\eg trusted authorities) that defeats decentralization.

To address the second challenge, some centralized \textit{relay networks} has been applied (FIBRE\footnote{http://bitcoinfibre.org/} and Falcon\footnote{https://www.falcon-net.org/} for Bitcoin) or proposed~\cite{klarmanbloxroute} to reduce the information (\eg blocks) propagation latency. However, all the existing relay networks (either centralized or distributed) for blockchains have paid little attention to the \textit{system dynamics} (which can be categorized into churn, \ie users leave and join the system), and \textit{network heterogeneity} (\ie relayers have different and time-varying bandwidth capacities). Without addressing the system dynamics and network heterogeneity, the performance of relay networks may not only be far from optimal, but it may also be detrimental in practice.

In this paper, we present \fissionb, which achieves both high system throughput and fast confirmation times without compromising security and decentralization. \fission improves upon the scalability limitations in several ways. First, we propose an Eager-Lazy pipeling model that separates every atomic transaction into two successive and independent sub-transactions, which can be processed separately without any communication cost (which significantly deteriorates scalability with sharding based solutions). This enables optimal parallelization while maintaining consistency. Second, we propose an adaptive partitioning mechanism that groups sub-transactions into different partitions according to the latest transaction volumes. As a result, an optimal tade-off between system throughput and confirmation times can be achieved. Last but not least, we employ a PoS-based consensus protocol to reach consensus on transactions in every partition with high security guarantees while maintaining decentralization. We investigate and prove the effect of system activity and honest threshold to the security guarantees, and then propose an online algorithm that enables consensus to be reached in a distributed and non-interactive manner.

\fission employs a hybrid network to minimize the information propagation latency, as a result, both system throughput and confirmation times can be significantly improved. The proposed hybrid network consists of a \textit{relay network} and a \textit{P2P network}. To achieve a stable and near-optimal the information propagation latency in the presence of unpredictable system dynamics and network heterogeneity, we propose an \textit{ultra-fast} probabilistic relayer selection algorithm which converges in $O(\loglog N^{\text{r}})$ steps, where $N^{\text{r}}$ is total number of nodes in the relay network. The P2P network of \fission  complements the relay network such that blocks and even the whole blockchain data can be retrieved from nodes, thus dramatically increases the data availability. We propose an online data retrieval strategy that enables any node to retrieve information from both the relay network and the P2P network in an efficient and cost-effective manner.

The rest of paper is organized as follows. Section~\ref{sec:overview} provides an overview of \fission fundamentals, including the data structures and cryptographic technologies used in \fission blockchain.  We then provide the details of the Eager-Lazy transaction model, the adaptive partitioning mechanism and the consensus protocol in Section~\ref{sec:ppap}. In Section~\ref{sec:relay}, we investigate the problem of minimizing the information propagation, and describe our relay selection algorithm, followed by the convergence analysis of the distributed algorithm. Section~\ref{sec:p2p} details the P2P storage layer, the proposed data retrieval strategy and its efficiency. Section~\ref{sec:conclusion} will conclude the paper.

\section{Overview of \fission}\label{sec:overview}
\fission consists of three layers, each of which is designed to fulfill the requirements of three core components of a blockchain: 1) the computation layer than enables all the nodes agree on the common view of a blockchain in presence of Byzantine failures, 2) the network layer that delivers transactions and blocks to all the nodes to create consensus, and 3) the storage layer that form an append-only, temper-proof distributed ledger with cryptographic data structures.

The native currency in \fissionb, denoted by \fitb, is an \textit{utility} token that enables users to participate in consensus and pay for the transaction processing, smart contract execution, \etc. Without loss of generality, any reference to amount, value, balance or payment in \fission should be counted in \fitb.

\subsection{Definitions}

\subsubsection{Account} 
\fission is a account-based blockchain. Each account has a pair of private and public keys. The private key, denoted by $sk$, is generated by secp256k1 curve~\cite{sec20002}  and it should be always kept secret. The public key, denoted by $pk$, is derived from the private key with Elliptic Curve Digital Signature Algorithm (ECDSA) \cite{johnson2001}. The public key is also referred to as the \textit{address} of an account, and it can be safely shared in public as it is almost impossible to derive the private key from a public key. One user may control many accounts, but only one public key may exist per account.

\subsubsection{Transaction} 
A transaction is essentially a digitally signed message, where the data includes: 1) the transaction type, 2) a token transfer from one account (\ie the sender) to another (\ie the receiver), 3) a scalar value to be transfered from the sender, 4) a nonce that indicates the number of transactions sent by the sender, 5) a hash of the additional data that can be stored in the P2P storage layer, and 6) a signature of the transaction that used to determine the sender of the transactions. The unique identifier of a transaction is the SHA3-256 hash of the transaction.

\subsubsection{Block} 
A block in \fission consists of 1) a body that contains a list of ordered transactions to be confirmed, and 2) a header that contains Merkle root arrays of these transactions, \textit{account table}, and \textit{transaction logs}, as well as the metadata (\eg signatures, votes) needed for the consensus. The account table is a mapping between accounts and their states which contains the information like balances, nonces, \etc., and the transaction logs contains the post-transaction states and logs created through execution of these transactions.

\subsubsection{Chain}
\fission can be viewed as a transaction-based state machine where every block contains a set of \textit{states} that can include such information as account balances, transaction logs, \etc. To apply the novel Eager-Lazy pipeling model (see Section~\ref{sec:eager-lazy}) that significantly improve the scalability, \fission introduces two types of blocks that are appended to the blockchain in an alternating manner: 1) the \textit{interim} block that confirms the eager sub-transactions, and 2) the \textit{main} block that confirms the corresponding lazy sub-transactions.

\subsubsection{Merkle root array} Every block has a set of Merkle root arrays (\eg \texttt{txRoot[]}), each of which is an array of Merkle roots of all the \emph{shards}. A Merkle tree is a binary tree in which every leaf node is labeled with the hash of a data block and every non-leaf node is labeled with the cryptographic hash of the labels of its child nodes. 

%The proposed consensus protocol, \textsf{PPAP} enables each shard to reach the consensus on each Merkle Root Array for every new block in parallel, thus significantly improve the throughput of the system.

\subsubsection{Shard} A shard is a subset of transactions and  states that are grouped by applying modulo function on their public keys. Let $pk_i$ denote the public key of node $v_i$, then all the transactions sent by $v_i$ and the states of $v_i$ will be assigned to the shard with index $j\equiv pk_i \mod N_{\text{shard}} $, where $N_{\text{shard}}$ is the total number of shards. The value  of $N_{\text{shard}}$ is contained in every block header, and it is increased based on re-sharding algorithm which is detailed in Section~\ref{sec:ap:re-sharding}.

\subsubsection{Partition} A partition consists of one or more shards. \fissions consensus is operating at the partition level such that a small group of nodes will be randomly selected (with probability proportional to their stakes) for each partition to reach a consensus on all the transactions in shards of the partition. Similarly, as detailed in Section~\ref{sec:ap}, shards are assigned to partitions by applying modulo function on their shard indexes.

\subsection{The Computation Layer}\label{sec:overview:comp}
The goal of the computation layer is to verify transactions, reach consensus on new blocks, and append blocks to the blockchain. To improve the system throughput, sharding, a commonly used scalability scheme in databases decades ago, has been applied to parallelize transaction processing via splitting the overheads of operation among multiple, smaller groups of nodes (\ie shards). However, all existing sharding-based solutions are far from optimal, as they impose a huge burden for the network due to cross-shard communication and synchronization. This significantly deteriorates both the system throughput and transaction confirmation times. 

To overcome the limitations of scalability, we propose a novel Eager-Lazy pipeling model, an adaptive partitioning mechanism and a PoS-based consensus protocol in \fissions computation layer, which significantly improve the scalability while maintaining security and decentralization. 

\subsubsection{Eager-Lazy pipeling model} 
We propose a Eager-Lazy pipeling model (see Section~\ref{sec:eager-lazy}) that maximizes the parallelization by block pipeling without cross-shard communication. The basic idea is to separate each \textit{atomic} transaction into two types of independent sub-transactions that can be processed sequentially while ensures consistency of the transaction. As a result, \fissions blockchain use two types of blocks that are appended to \fissions blockchain in an alternating manner such that each block contains a set of confirmed sub-transactions with the same type. 

\subsubsection{Adaptive partitioning mechanism}
To the best of our knowledge, none of existing sharding-based solutions adapts the time-varying transactions volumes. More specifically, all the transactions will be distributed to a fixed number of shards. If the transaction volume decreases (during non-peak hours), shards will be underutilized in terms of computation, but both the intra-shard and the inter-shard communication cost remain the same. To this end, we propose an adaptive partitioning mechanism (see Section~\ref{sec:ap}) that can accommodate transaction volumes by dynamically grouping shards in \textit{partitions} and processing transactions in partition basis. As a result, shards' resources can be optimally utilized and the network burden can be minimized.

\subsubsection{\ppap consensus protocol} 
Designing a consensus protocol for a scalable blockchain is very challenging. First, the consensus protocol should be able to avoid Sybil attacks~\cite{douceur2002sybil} - a common attack in open, decentralized environments where an \textit{adversary} (\ie a malicious user) can create multiple identities to influence the protocol. Moreover, the consensus protocol must be scale to high transaction volumes in an energy-efficient way. To address these challenges, we propose probably secure PoS-based consensus protocol, named Parallel Proof-of-Active-Participation (\ppapb, see Section~\ref{sec:ppap}), where \textit{active} (or online) participates (\ie nodes) are randomly selected (with probability proportional to their tokens) reach the consensus on partitions in parallel. We investigate the relationship between security guarantees of \ppap and the system activity and honest threshold.

\subsection{The Hybrid Network}
Due to the huge communication cost of broadcasting messages from one node to all the other nodes in the system and reaching consensus in partitions, we propose a hybrid network to minimize the information propagation latency, which is defined as the combination of transmission time and the local verification of the message. 

As shown in Fig.~\ref{img:hn}, the proposed hybrid network consists of : 1) a relay network, which consists of a small set of \emph{relay nodes} that forward messages (\eg transactions, blocks, \etc) in a many-to-many manner, and 2) a P2P network, where every node only broadcast its own messages to its neighbors, while broadcasting or forwarding the hashes of messages signed by other nodes in a gossip manner.
\begin{figure}[h!]
	\includegraphics[width= \columnwidth]{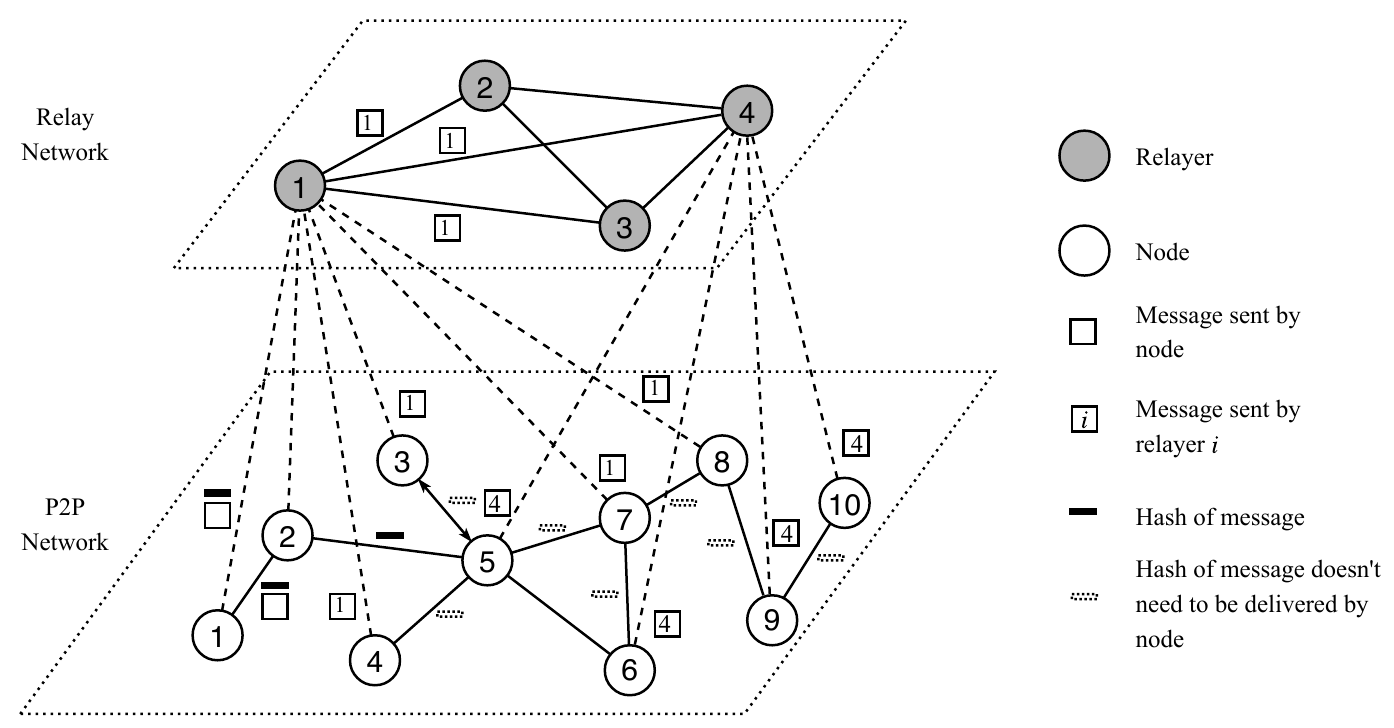}
	\caption{\fissions Hybrid Network.}\label{img:hn}
\end{figure}

\subsubsection{The relay network} 
The purpose of the relay network is to minimize the information propagation latency via reducing the number of hops that messages traverse before they reach all the online nodes. The relay network consists of a set of \emph{inter-connected } \textit{relayers} (servers with high network capacity and good hardware specifications) such that every relayer is \emph{incentivized} (via collecting a portion of transaction fees and block rewards) to deliver every message to a large number of nodes simultaneously. 

Specifically, each node will select one relayer to broadcast messages, and update, synchronize its local copy of the blockchain. To avoid garbage messages that may overwhelm the relay network, every relayer must validate messages (based on their signatures) before relaying them to the whole network. Furthermore, the relay network will forward all the hashes of messages, but only deliver the corresponding messages when nodes request. Due to the system dynamics that incurs dynamic and unpredictable loads on the relayer, it is plausible that nodes can effectively select or change their relayers, such that the overall system will not be vulnerable to those relayers that are under-performing or under DDoS attacks. 

Instead of deploying the relay network and managing the relayer selection in a centralized way, we design and implement the relay network in a \emph{strongly distributed setting}. Specifically, \textsf{Fission} enables every node to \emph{selfishly}  independently and concurrently select relayer that maximizes its own profit (\ie lower propagation latency) without any central control. We prove in Section~\ref{sec:relay:prs} that the expected information propagation latency is minimized once the loads on relayers are balanced proportional to their capacities, in which case a \emph{Nash equilibrium} is achieved. 

It is important to note that without accommodating the system dynamics and network heterogeneity in the relay network, the performance of any relay selection strategy may be detrimental in practice. To this end, we then propose a game-theoretic distributed algorithm, named Probabilistic Relay Selection (PRS) that leads the system to $\epsilon$-Nash equilibrium in an \textit{ultra-fast}\footnote{We follow the common use of the superlative ``ultra'' for double-logarithmic bounds~\cite{cohen2003scale,fountoulakis2012ultra}} $O(\loglog N^{\text{r}} )$ convergence time from any prior system state (\ie load distribution in the relay network), where $N^{\text{r}}$ is the number of relayers.

\subsubsection{The P2P network}
\fissions P2P network is based on the Kademlia DHT~\cite{maymounkov2002kademlia}, which is designed to be an efficient means for storing and finding content in a P2P network. Unlike other blockchains like Ethereum, where the Kademlia DHT is used only for node discovery, \fission fully utilizes the DHT-based P2P network as a highly available \textit{distributed storage system}, similar to IPFS~\cite{benet2014ipfs}. 

Similar to other blockchain P2P networks, transactions are disseminated in \fissions P2P network in a gossip-like manner. Besides, all the nodes will send their transactions to their relayers who help them to reach the whole network in 2 hops. Unlike existing blockchains, \fissions nodes do not forward blocks to their neighbors, although the block hashes are gossiped in the P2P network. In such a way, block information (\ie block hash) can be propagated (with help of the relay network) to the whole network as soon as possible.

Upon receiving a block hash, a node will choose to retrieve the block if it is selected as a committee member for the block. Otherwise it will ignore it (to effectively utilize the network resource). In \fissionb, every node can retrieve a block (based on its hash) from a content provider, defined as the node that has a local copy of data, with a lower fee, or from the relay network with a higher fee. In addition, if a node is a committee member for a block, it needs to retrieve the block within a \textit{latency constraint} (measured in seconds), otherwise it may miss the voting process, and as a result, lost the chance to get the block rewards. To minimize the cost of data retrieval with time constraints, we propose an online and light-weight  data retrieval strategy that achieves a good trade-off between performance and complexity.

\section{Eager-Lazy Pipeling}\label{sec:eager-lazy}
Although sharding promises to improve the throughput and reduce per-node processing and storage requirements, existing sharding-based blockchains still require a linear amount of communication per transaction (to sync up every transaction among different shards), and thus attain only partially benefits of sharding. 

To this end, we introduce a novel Eager-Lazy pipeling model to improve the system throughput via transaction pipeling. Specifically, each atomic transaction in \fission is divided into two successive and independent sub-transactions: 1) the \textit{eager} sub-transaction (\eg withdrawing tokens from the payer), and 2) the corresponding \textit{lazy} sub-transaction (\eg depositing tokens to the payee), as shown in Fig.~\ref{img:eager-lazy}. 
\begin{figure}[h!]
	\includegraphics[width=0.9 \columnwidth]{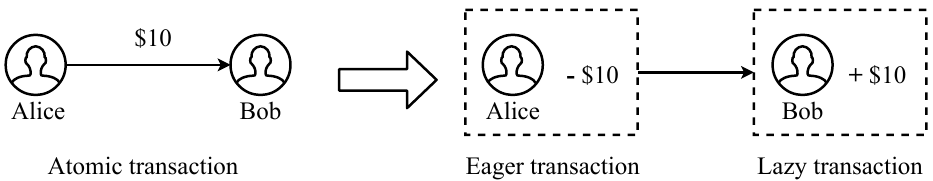}
	\caption{Transform between an atomic transaction and its corresponding eager and lazy sub-transactions.}\label{img:eager-lazy}
\end{figure}

The proposed Eager-Lazy pipeling overcomes the limitations of sharding-base solution via avoiding \ie high cross-shard communication. \fission confirms each atomic transaction via confirming the two sub-transactions \textit{separately} and \textit{independently} in the two successive blocks. In other words, all the eager sub-transactions will be confirmed in a block, referenced by another  block containing all the corresponding lazy sub-transactions. As a result, \fission introduces two types of blocks: 1) the \textit{interim} block that stores all the confirmed eager sub-transactions, and 2) the  \textit{main} block that stores all the confirmed lazy sub-transactions.
\begin{figure}[ht!]
	\includegraphics[width=0.9 \columnwidth]{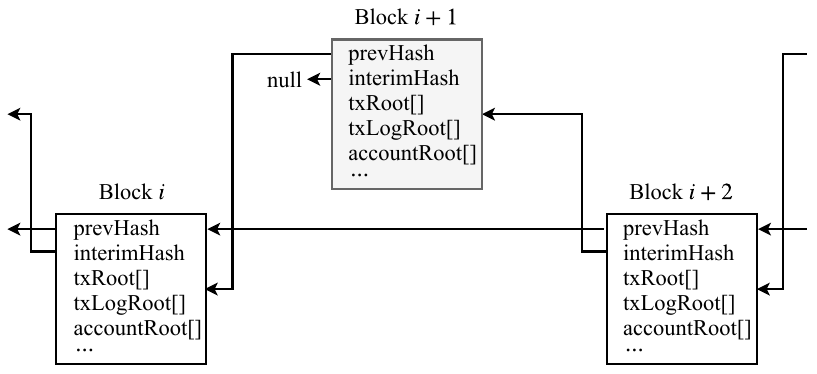}
	\caption{\fissions blockchain structure.The interim and main blocks are appended alternatively.}\label{img:blockchain}
\end{figure}

\fission blockchain proceeds in fixed time periods called \emph{epochs}. In every epoch, a new block that contains a set of sub-transactions is appended to the blockchain. Therefore, the blockchain consists of a sequence of concatenated blocks $<B_0, B_1, \cdots>$, where $i \geq 0$ indicates the epoch number. As shown in Fig.~\ref{img:blockchain}, blocks are appended to the blockchain in an \textit{alternating} way such that $B_i$ is a main block if and only if $i\mod  2 = 0$, otherwise it is an interim block.  

\subsection{Micro Block}\label{sec:eager-lazy:micro}
Note that all the sub-transactions will be processed in parallel with the proposed adaptive partitioning mechanism (explained in Section~\ref{sec:ap}). Upon reaching a consensus, a \textit{micro} block that contains all the confirmed sub-transactions in a partition will be generated and broadcast to all the partition committee members. Otherwise, the corresponding transactions will be processed in the next round.
\begin{figure}[ht!]
	\includegraphics[width=0.8 \columnwidth]{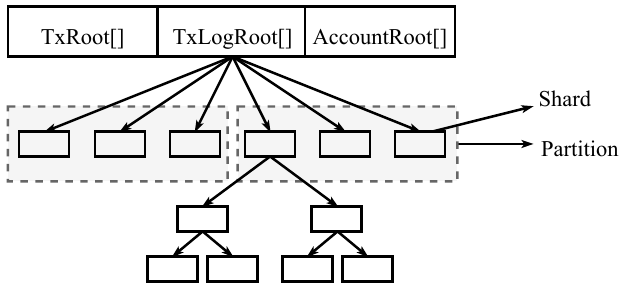}
	\caption{\fissions block structure.}\label{img:block}
\end{figure}
\subsection{Interim Block} 
An interim block is a combination of all the micro blocks that contains the eager sub-transactions. The transactions, the account table, and the transaction logs are distributed in shards (not partitions), simply based on the public keys of the senders. As shown in Fig.~\ref{img:blockchain} and Fig.~\ref{img:block}, each interim block contains three Merkle root arrays, each of which is a array of Merkle root of each shard. It should be noted that not all the eager sub-transactions in a partition can be processed within $\Delta_{\texttt{micro}}$ (\ie, the time constraint for a micro block) due to both the  heterogeneous complexity of transactions (\eg some of them are one-to-many fund transfers, or smart contracts), and the heterogeneous computation capacities of committee members in a partition. An interim block will be generated by the (interim) block committee by combining all the micro blocks from all the partitions within a time constraint $\Delta_{\texttt{interim}}$. Once a consensus is reached, the interim block will be appended to the blockchain, otherwise an empty interim block will be appended.

\subsection{Main Block} 
Once an interim block is appended to the blockchain, a (main) block committee will be selected to reach a consensus on a main block that contains the corresponding lazy sub-transactions within a time constraint  $\Delta_{\texttt{main}}$. Considering that lazy sub-transactions are basically credit operations based on logs of the eager sub-transactions (\eg \texttt{txLogRoot[]}) in the interim block, they are supposed to be processed much faster than eager sub-transactions. Furthermore, all the corresponding lazy sub-transactions in the interim block need to be confirmed in the main block which reference the interim block by \texttt{interimHash} in the block header. Otherwise, an empty main block will be appended.

\section{Adaptive Partitioning}\label{sec:ap}
Allocation of resources to processing transactions in parallel is challenging for blockchains of all sizes. Carving out or allocating parts of the system to run tasks without interfering with each other is commonly referred to as ``\textit{partitioning}''. Partitioning, in general, is the ability to divide up system resources into groups of parts in order to facilitate particular functions. Recently, some partitioning scheme have been proposed which partition a blockchain into persistent or static partitions (\ie shards). However, the capacities of the parallel partitions would be \textit{underutilized} if only a fixed number of transactions were to be proceeded on the fixed number of shards. In order to perform efficient scheduling of resources, we propose an \textit{adaptive partitioning mechanism} to process transactions parallelly in different \textit{partitions}, instead of shards. System performance can be further improved by adaptively determining the number of shards allocated to a partition based on the transaction volumes.

Each partition consists of one or multiple \textit{shards}, each of which consists of eager or lazy sub-transactions, account table, and transaction logs assigned to it based on the public keys of senders and accounts, respectively. The sub-transactions in each partition will be processed and a micro block (see Section~\ref{sec:eager-lazy:micro}) will be generated and agreed by a small group of nodes, called \textit{partition committee}, randomly selected using \ppap protocol based on the latest block and the partition index. Specifically, let $j$ be a shard index, then it will be assigned to a partition with index $k\equiv j \mod N^p$, where $N^p \leq N^s$ is the total number of partitions. The value of $N^p$ is contained in every block header, and it is adjusted based on the latest transaction volumes.

\subsection{Number of Partitions}
The purpose of the proposed adaptive partitioning mechanism is to reach consensus on the final block as soon as possible respecting the fluctuation of transaction volumes. The larger the number of transactions is, more partitions are required such that partitions are not overloaded. In \fissionb, a simple auto-scaling strategy is applied to determine the number of partitions as follows: Let $N^e_i$ and $N^p_i$ be the number of confirmed eager sub-transactions, and the number of partitions in the latest block $B_i$, respectively. The number of partitions for the next block $B_{i+1}$, denoted by $N^p_{i+1}$ is derived as following:
\begin{equation}\nonumber
N^p_{i+1} =
\begin{cases}
N^p_{i} + 1, & \text{if } \frac{N^e_i}{N^p_{i} } \geq \delta N^e_{\max} \\
N^p_{i}  - 1, & \text{if } \frac{N^e_i}{N^p_{i} } \leq (1-\delta)N^e_{\max} \\
N^p_{i} ,       & \text{otherwise}
\end{cases},
\end{equation}
where $N^e_{\max}$ is the maximum number of sub-transactions a partition can  process, and $\delta \in (0, 1)$ is a scale factor. Both $N^e_{\max}$ and $\delta$ are pre-determined in \fissionb, and they can be adjusted for the best practice.

\subsection{Re-Sharding}\label{sec:ap:re-sharding} 
Note that the state and storage information of nodes are organized using Merkle tree structure, which is essentially a binary tree. It is straightforward and efficient to increase the number of shards by splitting a Merkle tree into two or more Merkle trees. Re-sharding will be triggered automatically if there are a number of successive main blocks, say $(B_i, B_{i+1},\dots, B_{i+N^{rs}})$, we have $N^p_{j} = N^s_j$, for all $j\in[i, i+N^{rs}]$, where $N^s_j$ denotes the number of shards in block $B_j$, and $N^{rs}$ is a pre-determined number.

\section{PPAP Consensus}\label{sec:ppap}
Byzantine agreement protocols, \eg PBFT\cite{castro2002practical},  have been used to sync up states among a relatively small group of servers, and it has been used in other blockchains to reach consensus. However, PBFT requires a fixed number of servers acting like bootstrapping nodes thus those blockchains may be vulnerable to Sybil attacks. Moreover, it does not scale to a large number of nodes (say over 100,000 nodes) who are participating in the consensus process. 

To scale the \textsf{Fission}'s consensus process to many online nodes, a small group of nodes are selected (see Section~\ref{sec:ppap:committee-selection}) as the committee for each partition and block at each epoch. Once the nodes are selected and their voting power are determined, a Byzantine agreement protocol is executed by every selected node (\ie every committee member) to reach a consensus on the new block.

\textsf{Fission}'s consensus process is purely decentralized such that every node can be selected as a committee member. To avoid Sybil attacks, where adversaries may create millions of nodes with negligible tokens to increase the probability of being selected, and thus increase the probability of appending malicious blocks, \fission enables nodes to be selected randomly with probability proportional to their tokens, and their voting power are proportional to their tokens. As a result, \fissions consensus protocol is a PoS-based protocol. 

\subsection{Active Participation}
Furthermore, \fissions consensus protocol relies on the \textit{active participation} of nodes, who improve the security and performance of the system by verifying transactions and voting for the new blocks to be appended to the blockchain. It is worth to note that all the blockchains require active participation from their users. Without the participation of users, the security of blockchains cannot be guaranteed. Therefore, \fissions consensus protocol is named as Parallel Proof-of-Active-Participation, \ppap in short, where the security of \fission is guaranteed by online nodes who are actively participating the transactions/blocks verification and block producing. 

\subsection{Committee Selection} \label{sec:ppap:committee-selection}
In \fissionb, committees are randomly selected from all the online nodes at each epoch $e$ in a non-interactive manner, with probability proportional to their tokens (\ie stakes). Let $V_e=\cup_{i\in \mathbb{Z}^{+}} \{v_i\}$ and $S_e=\cup_{i\in \mathbb{Z}^{+} } \{ s_i\}$ denote the set of all the active nodes, and the set of their tokens at epoch $e\in \mathbb{Z}^{+}$, respectively. Consider every node, say $v_i$ performs a Bernoulli trial on every token it has. Let $p$ be the probability of success in the Bernoulli trial. Then the probability of being selected is $1-(1-p)^{s_i}$, and the expected voting power is $p\times s_i$. Therefore, the voting power of node $v_i$, defined as $o_i$, is a random variable follows Binomial distribution, \ie $o_i\sim \mathcal{B}(s_i, p)$. The cumulative distribution function (CDF) for $o_i$ is defined as the following:
\begin{align}\label{eq:ppap:vote-cdf}
F(k; s_i, p) &= \Pr(o_i \leq k) = \sum_{j=0}^k \binom{s_i}{j}p^j(1-p)^{s_i-j}
\end{align}

To generate the $o_i$ given $s_i, p$ for each $v_i\in V_e$, we use a simple algorithm based on the inverse transform method \cite{gentle2006random}, as shown in Fig.~\ref{img:inverse-cdf} and Algorithm~\ref{alg:ppap:cal-vote}.
\begin{figure}[ht!]
	\includegraphics[width=0.6 \columnwidth]{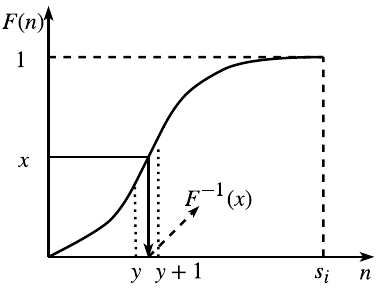}
	\caption{Calculating voting power of $v_i$ using inverse transform method. $x$ is random number uniformly distributed in $[0,1]$, and $y$ is the maximum integer such that $y \leq F^{-1}(x)$. Then the voting power of $v_i$ is $y$.}\label{img:inverse-cdf}
\end{figure}
Specifically, every node will generate a random number, say $x$ from a uniform distribution on the interval $(0,1)$ (Line 1 of Algorithm~\ref{alg:ppap:cal-vote}). Then it compares $x$ with $F(0; s_i, p)=(1-p)^{s_i}$ (\ie the probability that none of $s_i$ tokens is selected) and stop and return $0$ if $x$ is not greater, as shown in Lines 2-3 of Algorithm~\ref{alg:ppap:cal-vote}. The inverse CDF (\ie $F^{-1}(x)$) is used to calculate the voting power of $o_i$ given the random number $x$, and $s_i, p$.
\begin{algorithm}[ht!]
	\caption{Calculating voting power given ($s_i, p$).}\label{alg:ppap:cal-vote}
	Generate x from uniform distribution $X\sim \mathcal{U}(0,1)$ \; 
	\If{ $x \leq F(0; s_i, p)$ }{Return 0\;}
	Return $\sup \{y\in \mathbb{Z}^{+}: F(y) \leq F^{-1}(x) \}$\;
\end{algorithm}

Every node in \fission uses VRFs~\cite{micali1999verifiable} to generate random numbers to calculate their voting power. Because VRFs allows other node to verify those numbers efficiently. Furthermore, it enables nodes to calculate their voting power in a non-interactive and independent manner. Let $B_e$ denote the block at epoch $e$, and let $seed_{e}$ denote the seed information used by all online nodes to calculate their voting power for the next block $B_{e+1}$. Every node, say $v_i\in V_e$, will generate a tuple using VRF as follows:
\begin{equation}
(hash_i^e, \pi_i^e) =VRF(sk_i, seed^e, type),
\end{equation}
where $hash_i^e$ is a 256-bit long pseudo-random value that is uniquely determined by $sk_i$ (private key of node $v_i$), $seed^e$ (random seed of block $B_e$), and the type of committee (\eg partition or block committee), but is indistinguishable from random to any node that does not know $sk_i$. $\pi_i^e$ is the proof that anyone can verify that $hash_i^e$ is valid with the knowledge of $pk_i$, the public key of node $v_i$. Therefore, the random number that is used to calculate voting power of node $v_i$ given $s_i, p$ is $hash_i^e / 2^{256}$.

Considering that it is impractical to assume $V_e$ or $n=|V_e|$ is known as a prior, we choose $p$ as a pre-determined value which does not depend on dynamic, unpredictable information. Specifically, $p=\tau / K$, where $K$ is the total number of tokens, and $\tau$ is the expected number of tokens selected. Choosing a good $\tau$ is very important as it affects both the performance and the security guarantee of the consensus algorithm, and we will discuss it in Section~\ref{sec:ppap:security}. 

\subsection{Byzantine Agreement}\label{sec:ppap:ba}
\fissions Byzantine agreement works roughly as follows: for each block committee, a \textit{block proposer} will be elected to propose a new block. Specifically, at the beginning of every epoch $e$, every block committee member $v_i\in V_e$ generates a ticket as follows: 
\begin{align}
ticket_i=VRF(sk_i, seed^e,``\text{leader}")
\end{align}
Then block committee members then gossip their tickets with each other for a time $\Delta_{\text{leader}}$, after which they elect the valid ticket with the lowest value they have seen and accept the corresponding node as the block proposer, \ie \textit{leader}. If this leader is or becomes unavailable, leadership passes to the next node in ascending order of tickets. Then the block proposer will broadcast the hash of the candidate block to the whole network. Upon received this hash, every block committee member will retrieve the block (from its relayer or other nodes) before verifying all the transactions in the block, and then broadcast its votes with its signature. Once the block proposer receive enough votes (\ie $\geq \theta$, see Section~\ref{sec:ppap:security} for more details) on the new block, it will gossip the confirmed block to the whole network, where each node will append this block to its local copy of the blockchain.

\subsection{Security Guarantee}\label{sec:ppap:security}
Let $S_h$, and $S_a$ denote the total number of tokens owned by the online honest nodes, and adversaries (\ie malicious users), respectively. Let $h$ be the \textit{honesty threshold}, indicating the fraction of tokens that are owned by honest nodes, \ie $S_h = h(S_h+S_a) $.  Let $\alpha$ be the \textit{activity of system}, defined as the fraction of tokens that are owned by the online nodes. Therefore, we have $(S_h+S_a) = \alpha K$. 

Let $X_h$ and $X_a$ denote the number of tokens selected (via executing Algorithm~\ref{alg:ppap:cal-vote}) from honest nodes, and adversaries, respective. As discussed in Section~\ref{sec:ppap:committee-selection}, we have $X_h \sim \mathcal{B}(S_h, p)$, and $X_a \sim \mathcal{B}(S_a, p)$, that is both $X_h$ and $X_a$ follow the binomial distribution. Note that $S_h$ and $S_a$ are \textit{sufficient large}\footnote{Which means $S_hp > 5, S_ap > 5$}, thus both $\mathcal{B}(S_h, p)$ and $\mathcal{B}(S_a, p)$ can be approximated with Normal distribution as follows:
\begin{equation}
X_h \sim \mathcal{N}(\mu_h, \sigma_h^2), \mu_h=S_hp,  \sigma_h^2 = S_hp(1-p)
\end{equation}
\begin{equation}
X_a \sim \mathcal{N}(\mu_a, \sigma_a^2), \mu_a=S_ap,  \sigma_a^2= S_ap(1-p)
\end{equation}

Define $Y=X_h - 2X_a$, and it also follows Normal distribution, \ie $Y\sim N(\mu_y, \sigma_y^2)$, where $\mu_y= (3h-2)\alpha\tau$, and $\sigma_y^2 = (4-3h)\alpha\tau$. To ensure the consensus is reached on valid blocks, the fraction of selected tokens are owned by adversaries should be less than 1/3. Therefore, the probability $\Pr(Y\leq 0)$ should be negligible given $h$ and $\alpha$. We transform $Y$ to a random variable $Z$, which follows a standard normal distribution,\ie
$Z\sim N(0,1)$, then we have the following:
\begin{equation}
\begin{split}
\Pr(Y\leq 0) &= \Pr\left(Z\leq -\frac{\mu_y}{\sigma_y}\right) \\
&\approx \Pr\left(Z \leq -\frac{(3h-2)\alpha\tau}{\sqrt{(4-3h)\alpha\tau}} \right).
\end{split}
\end{equation}
Considering $\Pr(Z\leq -6.36) \leq 10^{-10}$, hence we have
\begin{equation} \label{eq-ppap-tau}
\tau > \frac{40.5(4-3h)}{(3h-2)^2\alpha}
\end{equation}

To ensure that a consensus on a valid block can be reached in a non-interactive manner, the fraction of weights (\ie votes), denoted by $\theta$, needs to be pre-determined such that a block is valid with high probability (\ie $> 1-10^{-10}$) if there are more than $\theta \tau$ votes are received on the block. Hence $\Pr(X_a\geq \theta\tau)$ and $\Pr(X_h < \theta\tau)$ should be negligible, even the activity of system (\ie $\alpha$) is relative low, as shown below:
\begin{equation}
\begin{split}
\Pr(X_a \geq \theta\tau) &= \Pr\left(Z \geq \frac{\theta\tau - \mu_a}{\sigma_a}\right) < 10^{-10} \\
\Pr(X_h < \theta\tau) &= \Pr \left( Z < \frac{\theta\tau -\mu_h}{\sigma_h} \right) < 10^{-10},
\end{split}
\end{equation}
where $Z\sim N(0,1)$ is a random variable follows standard normal distribution. Considering $\Pr(Z > 6.36) < 10^{-10}$ and $\Pr(Z < -6.36) < 10^{-10}$, hence we have
\begin{equation}\label{eq-ppap-theta}
(1-h)\alpha + 6.36\sqrt{\frac{(1-h)\alpha}{\tau}} < \theta  < h\alpha - 6.36\sqrt{\frac{h\alpha}{\tau}}
\end{equation}

Note that one of most secure and well-adopted blockchain, Bitcoin, is still vulnerable to attacks (\eg selfish mining\cite{eyal2018majority}) if adversaries control over $25\%$ of hashrate. We assume that the honest threshold of the system, $h$, should be larger than $0.75$. Otherwise \fission will be vulnerable to attacks by adversaries. As the right part of Eq.~\ref{eq-ppap-tau} is a monotone and non-increasing function over $h$, we have $\tau > \frac{1134}{\alpha}$ given $h\geq 0.75$. To satisfy Eq.~\ref{eq-ppap-theta} under the same assumption of $h\geq 0.75$, $\tau$ should be relative large to be adaptive to different activity of system, \ie $\alpha$. However, a larger $\tau$ results in a higher communication cost, thus $\tau$ needs to be determined considering the trade-off between the communication cost and security. Specifically, based on Eq.~\ref{eq-ppap-theta} and Eq.~\ref{eq-ppap-tau}, we have the lower bound of $\theta$ given $h\geq 0.75$ and $\alpha \in [0,1]$:
\begin{equation}
\begin{split}
\theta &> (1-h)\alpha + 6.36\sqrt{\frac{(1-h)\alpha}{\tau}} \\
&> \frac{1}{4} + \frac{3.18}{\sqrt{\tau}}.
\end{split}
\end{equation}
In \fission we set $\tau=5000$, and $\theta=0.3$. Therefore, to satisfy the right inequation of Eq.~\ref{eq-ppap-theta}, we have
\begin{equation}
h\alpha - 6.36\sqrt{\frac{h\alpha}{5000}} > 0.3,
\end{equation}
and thus $\alpha > 0.53$. Therefore, given $h \geq 0.75$, $\alpha > 0.53$, and $\tau=5000$, a consensus on a valid block can be reached \whp if over $1500$ votes are received.

\section{The Relay Network}\label{sec:relay}
The relay network of \fission is a mesh network consists of a group of relayers who contribute their resources (\eg bandwidth, memory and computation) to propagate messages (\eg transactions, blocks, \etc). With relayers, all messages can be delivered within at most 3 hops, thus the information propagation latency is significantly reduced.

The relay network in \fission is implemented in a strong practical and distributed settings: 1) there are no oracle or centralized authorities managing or controlling the relayers, 2) each relayer has no knowledge (\eg load, number of connections, \etc) about other relayers, and 3) relayers have heterogeneous bandwidth capacities and hardware specifications, and both can vary over time. 

To achieve the optimal propagation latency in the distributed settings, nodes need to change their relayers from overloaded ones to underloaded ones independently and concurrently (see Section~\ref{sec:relay:relay-update}). However, It can be observed that if all nodes behave greedily at each step (\ie, they select those relayers with minimum load deterministically), the system may not converge to a steady state \cite{Berenbrink:2006:DSL:1109557.1109597}. To guarantee that the system will converge rapidly, some probability rules need to be imposed for all nodes when they decide to change their strategies at each step. In \fission, we proposed a \textit{probabilistic} relayer selection algorithm, called PRS (see Section~\ref{sec:relay:prs}), which leads the system converge to an $\epsilon$-Nash equilibrium (see Section~\ref{sec:relay:min-delay}) where a near-optimal propagation latency is achieved. 

\subsection{Relay Selection and Update}\label{sec:relay:relay-update}
After joining the system, every node $v_i\in V$ will randomly select a relayer, denoted by $m_j\in M$ with probability proportional to their advertised bandwidth capacities, and then it will invoke a periodic relay update process to choose a \textit{better} relayer every $\Delta_\texttt{relay}$ seconds (the update interval). 

At each step, node $v_i$ randomly selects a different relayer, defined as the relay candidate, in the system. Node $v_i$ will potentially replace the current relayer with the relay candidate using Algorithm~\ref{alg:relay:prs}, detailed in Section~\ref{sec:relay:prs}. It is worth noting that every node will incur a certain \textit{load} on its relayer. The \textit{load} of $m_j \in M$, denoted by $l_j$, is defined as the number of nodes that choose $m_j$ as their relayer, and the \textit{load ratio} of $m_j$, denoted by $r_j$, is defined as $l_j/u_j$.

\subsection{Minimizing Propagation Delay}\label{sec:relay:min-delay}
The key performance metric for information propagation in bandwidth constrained networked environments is the averaged propagation latency. In \fission, every node will send the message (transactions and blocks) and its hash to the relayer, who then forward the message to all the connected nodes who did not receive the message yet. Therefore, the distribution of loads on the relayers will significantly affect the propagation delay.

Considering that relayers have heterogeneous and time-varying resources in practice, we investigate the problem of minimizing expected information propagation delay respecting the heterogeneous bandwidth constraints, as shown in Theorem~\ref{theorem:relay}. We then propose a game-theoretic relay selection algorithm, \ie PRS, to minimize the averaged information propagation delay in a fully distributed manner.

\begin{theorem}\label{theorem:relay}
	Let $V=\cup_{i\in \mathbb{Z}^+} \{ v_i \}$ be the set of online nodes. Let $M=\cup_{i\in \mathbb{Z}^+} \{ m_i \}$ and $U=\cup_{i\in \mathbb{Z}^+} \{ u_i \}$ be the set of relayers and their advertised upload capacities, respectively. Given any distribution of load ratios on the relayers, $\mathbf{r}$, the expected averaged information propagation delay, denoted by $\E(\mathbf{r})$, is minimized if and only if $\forall r_i\in \mathbf{r}, r_i=\overline{r}=\frac{|V|}{|U|}$.
\end{theorem}
\begin{proof}
	As every node selects a relayer u.a.r., the number of messages to be propagated by a relayer is proportional to its load (\ie number of nodes that select it as the relayer). Therefore, the expected number of message to be propagated by the relayer $l_j/|V|$. For every message, $m_j$ is supposed to forward to $l_j$ nodes, thus the propagation delay incurred on relayer $m_j$ is $\frac{\overline{w} l_i^2}{u_i |V|}$, where $\overline{w}$ denotes the averaged message size. Furthermore, the problem of minimizing the expected propagation delay for any node $v_i\in V$ is formulated as the following
	\begin{align}
	\text{minimize} &: \hspace{10pt} \E(\mathbf{r}) = \frac{\overline{w}}{|V|} \sum_{m_i\in M}\frac{l_i^2}{u_i}, \\
	\text{s.t.} &: \hspace{10pt} \sum_{m_i\in M}l_i=|V|,
	\end{align}
	By the method of Lagrange multipliers, $\E(\mathbf{r})$ is minimized if $\frac{l_i}{u_i} = \frac{l_j}{u_j}, \forall m_i, m_j \in M$, and this completes the proof.
\end{proof}
Theorem~\ref{theorem:relay} implies that to minimize the expected propagation delay, the loads on the relayers must be balanced proportionally to their advertised upload bandwidth. To achieve a scalable system, our goal is to distribute and balance the loads on the relayers in a fully distributed manner. We then formulate the distributed load balancing problem as an \textit{asymmetric congestion game}, where each player (\ie node) in the game can change his/her strategy (\ie relayer) individually
and concurrently.

\subsection{Congestion Game Preliminaries}
The classical congestion games have been investigated for many years. A congestion game, denoted by $G$, can be defined as a tuple $(V,M,(A_i)_{i\in V},(f_e)_{e\in M})$, where $V$ denotes the set of \emph{players} and $M$ denotes a set of \emph{facilities}. Each player $i\in V$ is assigned a finite set of \textit{strategies} $A_i$ and a cost function $f_e$ is associated with facility $e\in M$.

To play the game, each player $i\in V$ selects a strategy $a_i\subseteq A_i$, where $A_i\subseteq M$ is the strategy set of player $i$. The strategy profile, denoted by $\mathbf{a} = (a_i)_{i\in V}$, is defined as a vector of strategies selected by all the players. Similarly we use the notation $\mathcal{A} = \times_{i\in V}A_i$ to denote the set of all possible strategy profiles. A congestion game is \emph{symmetric} if all the players have the same strategy set, \emph{i.e.}, $\forall i,j\in V, A_i=A_j$; otherwise it is \emph{asymmetric}. A congestion game is  \emph{weighted} if each player $i$ is specified a weight $w_i$. The cost of player $i$ for the strategy profile $\mathbf{a}$ is given by $c_i(\mathbf{a}) = f_{a_i}(\mathbf{a}, w_i)$.

The goal of each player in congestion games is to minimize her own cost without trying to optimize the global situation. That is, all players will try to lower their own cost by changing their strategies individually. A \emph{pure Nash equilibrium} is defined as a steady state in which no players have an incentive to change their strategies.
\begin{definition}[Pure Nash equilibrium~\cite{nisan2007algorithmic}]
	A strategy profile $\mathbf{a}\in \mathcal{A}$ is said to be a pure Nash equilibrium of $G$ if for all players $i\in V$ and each alternate strategy $a_i^{\prime}\in A_i$,
	\begin{align}
	c_i(a_i, \mathbf{a}_{-i}) \leq c_i(a_i^{\prime}, \mathbf{a}_{-i}),
	\end{align}
	where $\mathbf{a}_{-i}=(a_j)_{j\in V\setminus\{i\}}$ denotes the list of strategies of the strategy profile
	$\mathbf{a}$ for all players except $i$.
\end{definition}

Congestion games have a fundamental property that a pure Nash equilibrium always exists~\cite{fabrikant2004complexity}. To analyze convergence properties of congestion games (\emph{i.e.}, the time from any state to a pure Nash equilibrium), we introduce the definition of potential function for congestion games.
\begin{definition}[Potential function~\cite{even2005fast}]
	A function $\Phi: \mathcal{A} \rightarrow R$ is a potential for game $G$ if $\forall \mathbf{a}\in \mathcal{A}$, $\forall a_i,a_j\in A_i$,
	\begin{align}\label{eqn:po}
	c_i(a_i, \mathbf{a}_{-i}) \leq c_i(a_j, \mathbf{a}_{-i}) \Rightarrow \Phi(a_i, \mathbf{a}_{-i}) \leq \Phi(a_j, \mathbf{a}_{-i}).
	\end{align}
\end{definition}
Therefore, $\mathbf{a}^{*}$ is a Nash equilibrium if and only if
\begin{align}
\mathbf{a}^{*} = \argmin_{\mathbf{a}\in \mathcal{A}}\Phi(\mathbf{a}). \label{alg:nash}
\end{align}

An $\epsilon$-Nash equilibrium is an approximate Nash equilibrium, which is defined as a state in which no player can reduce her cost by a multiplicative factor of less than $1-\epsilon$ by changing her strategy.

We define a potential function  to associate the expected propagation delay with $\mathbf{r}=(r_1, r_2, \dots)$, which is the load ratio vector on the relayers:
\begin{equation}
\Phi(\mathbf{r})=\sum_{r_i\in \mathbf{r}}^{|M|}(r_i-\overline{r})^2,
\end{equation}
where $\overline{r}=|V|/|U|$ is the optimal load ratio. The potential function has the property that if node $v_i$ switches from a relayer with a high load ratio to another relayer with a low load ratio, $\Phi(\mathbf{r})$ will decrease accordingly. 

Note that every node is selfish and will change its strategy to lower its own cost at each step. This will result in \textit{Nash dynamics} -- \ie $\Phi(\mathbf{r})$ will fluctuate over time. To guarantee that the system converges to a steady state rapidly, a set of probability distributions over $A_i$ (i.e., strategy set) needs to be assigned to each node $v_i\in V$. That is, all the nodes' strategies are nondeterministic and are regulated by a probabilistic rule.

\subsection{Probabilistic Relayer Selection (PRS)}\label{sec:relay:prs}
We now present the probabilistic relay selection (PRS) algorithm based on an asymmetric congestion game. PRS enables each node to update its relayer as follows (shown in Algorithm~\ref{alg:relay:prs}). At each step, node $v_i$ that selects a relayer $m_j$ will contact a relay candidate $m_k$ with probability proportional to its advertised upload bandwidth. At the same time it finds the load and the capacity of $m_k$ (i.e., the relay candidate). Let $r_j$ and $r_k$ be the load ratios of $m_j$ and $m_k$, respectively. If $r_j > r_k$, then $v_i$ replaces $m_j$ with $m_k$ with a probability, denoted by $\Pr(j,k) $, shown as below:
\begin{align}\label{eq:relay:relay-selection}
\Pr(j,k) = \left\{\begin{array}{l l}
\frac{u_i}{|U|}\left( 1-\frac{r_k}{r_j} \right) \ \ \ & \mbox{if $r_j > r_k$,} \\
0 \ \ \ & \mbox{if $j\neq k, r_j \leq r_k$,} \\
1-\sum_{j\neq k}\Pr(j,k)                                             \ \ \ & \mbox{if $j=k$.}
\end{array} \right.
\end{align}
where $u_k$ is the advertised upload capacity of $m_k$, and $|U|$ is the total advertised upload capacities of all the relayers.

For very relayer $m_k\in M$, the probability of being selected as a candidate relayer is $\frac{u_k}{|U|}$. Our solution is to maintain an \emph{identifier space}, which can be managed either by a distributed scheme (e.g., using DHT), or by a centralized scheme (i.e., using a server to store all the identifiers). Each identifier consists of a relayer ID (\ie, its name or IP address) and a randomly generated hash key. To let a relayer to be chosen probabilistically proportional to its upload capacity, each relayer $m_k$ will maintain $\frac{u_k}{\mu}$ identifiers, where $\mu$ is a scalar value measured by Kbps. Therefore, at each step of the relayer update process, node $v_i$ will choose an identifier from the identifier space u.a.r. (uniformly at random), hence, the probability of relayer $m_k$ to be contacted by $v_i$, is proportional to its upload capacity $u_k$.
\begin{algorithm}[ht!]
	\DontPrintSemicolon
	\caption{Probabilistic Relayer Selection} \label{alg:relay:prs}
	$\forall v_i\in V$ that selects $m_i\in M$, do the following\;
	Select and contact $m_k\in M$. \; \label{alg:select}
	\While{$m_k$ is unreachable}{Go to line \ref{alg:select}\;}
	Let $r_j$ and $r_k$ be the load ratios of $m_j$ and $m_k$, respectively\;
	\If{$r_j > r_k$}{
		Replace $m_j$ with $m_k$ with probability $\left( 1 - \frac{r_k}{r_j} \right)$\;
	}
\end{algorithm}

An important criteria to evaluate distributed algorithms is the \textit{convergence time}, which is a measure of how fast the algorithm leads the system reaches a steady state. Otherwise system dynamics may significantly degrade the performance. We analyze the upper bound of the convergence time of PRS, where all the nodes behave simultaneously at each step.

Let $\mathbf{r}_t$ and $\mathbf{r}_{t+1}$ be the load ratio vectors of the current step $t$ and the next step $t+1$, respectively. Let $\mathbf{r}_t^i=l_i/u_i$ denote the load ratio of relayer $m_i$ at time $t$. Since we consider a highly dynamic and heterogeneous environment where nodes can join, leave and make transactions at will, we are especially interested to show how fast PRS can lead the system to converge to a $\epsilon$-Nash equilibrium. We observe that $\max\Phi_{\epsilon\text{-Nash}} = \frac{\epsilon^2m}{4}$, where $m$ is the total number of relayers. Hence we say that a $\epsilon$-Nash equilibrium is reached at time $t$ if $\Phi(\mathbf{r}_t)=O(m)$.

\begin{theorem}\label{theorem:plb:convergence}
	Let $m$ be the number of relayers, and let $\Phi(\mathbf{r}_t)$ be the potential function of the relayers at any step $t$. Then, the upper bound on the number of steps that $\Phi(\mathbf{r}_t)$ requires to reach a $\epsilon$-Nash equilibrium is $O(\log\log m)$.
\end{theorem}
\begin{proof}
	To better structure the proof, we introduce two supporting lemmas provided in the Appendix. Given any step $t$, we observe from Lemma~\ref{lemma:plb:exp} that the expected load ratio of any node at step $t+1$ is the optimal load ratio $\overline{r}$. Lemma~\ref{lemma:plb:var} provides an upper bound of the variance of load ratios at step $t+1$. Consider $\E\left[\Phi(\mathbf{r}_{t+1})\right]$, which is the expected value of the potential function at step
	$t+1$:
	\begin{align}
	\E\left[\Phi(\mathbf{r}_{t+1})\right]
	&=\E\left[\sum_{i=1}^m\left(\mathbf{r}_{t+1}^i-\overline{r}\right)^2\right].
	\end{align}
	Note that all the nodes behave independently at each step (\ie select a relayer with lower load ratio), hence $\mathbf{r}_{t+1}$ is independent of $\mathbf{r}_t$. It follows that
	\begin{align}
	\E\left[\Phi(\mathbf{r}_{t+1})\right]
	&=\sum_{i=1}^m\E\left[\left(\mathbf{r}_{t+1}^i\right)^2\right]
	- 2\overline{r}\sum_{i=1}^m\E\left[\mathbf{r}_{t+1}^i\right] + m\overline{r}^2.
	\end{align}
	
	Lemma~\ref{lemma:plb:exp} indicates that the expectation of each relayer $m_i$, $\E\left[\mathbf{r}_{t+1}^i\right]$, is $\overline{r}$ at each step, hence we have
	\begin{align}
	\E\left[\Phi(\mathbf{r}_{t+1})\right]
	&=\sum_{i=1}^m\E\left[\left(\mathbf{r}_{t+1}^i\right)^2\right] - m\overline{r}^2 \nonumber \\
	&=\sum_{i=1}^m\Var\left[ \mathbf{r}_{t+1}^i\right] \leq \left(m\Phi(\mathbf{r}_t)\right)^\frac{1}{2}.
	\end{align}
	
	We observe that the square-root function is concave. Therefore, by Jensen's inequality, the expected value of the potential function at step $t+1$ is
	\begin{align}
	\E\left[\Phi(\mathbf{r}_{t+1})\right]
	&\leq \E \left[ \left(m\Phi(\mathbf{r}_t)\right)^\frac{1}{2} \right] \leq \left(m \E \left[\Phi(\mathbf{r}_t)\right] \right)^\frac{1}{2}.
	\end{align}
	
	Since $\E\left[\Phi(\mathbf{r}_t)\right] > 0$ for any step $t\geq0$, we define a function $f(t)=\log(\E\left[\Phi(\mathbf{r}_t)\right])$, and we have the following
	\begin{align}
	f(t+1)\leq \frac{\log m}{2} + \frac{f(t)}{2}.
	\end{align}
	
	After $\lambda \in \mathbb{Z}$ steps from any step $t\geq0$,
	\begin{align}
	f(t+\lambda)
	&\leq \left( 1- \frac{1}{2^{\lambda}} \right)\log m + \frac{1}{2^{\lambda}}f(t) \nonumber \\
	&\leq \log m + \frac{1}{2^{\lambda}}f(t).
	\end{align}
	
	Therefore, the expected value of the potential function after $\lambda$ steps from any step $t\geq 0$ is
	\begin{align}
	\E\left[\Phi(\mathbf{r}_{t+\lambda})\right] \leq m\E\left[\Phi(\mathbf{r}_t)\right]^{2^{-\lambda}}.
	\end{align}
	
	The upper bound of $\Phi(\mathbf{r})$ is $O(m^2)$, which occurs in the case when every node $v_i$ selects one specific relayer$m_j\in M$. Hence, the upper bound of $f(t)$ is $O(\log m)$, \ie there exists an integer $a\in \mathbb{Z}$ such that $\sup f(t) \leq a\log m$. Therefore, it suffices to show that there exists a $\lambda=\lceil \log (a\log m) \rceil \sim O(\log\log m)$, such that $\frac{f(\lambda)}{2^{\lambda}}\leq 1$. Hence, for any $t\geq0$, it suffices to show that
	\begin{align}
	\E\left[\Phi(\mathbf{r}_{t+\lambda})\right] \leq 2m,
	\end{align}
	where $\lambda=O(\log\log m)$. By Markov's inequality, we have
	\begin{align}
	\Pr(\Phi(\mathbf{r}_{t+\lambda}) \geq 4m) &\leq \frac{\E\left[\Phi(\mathbf{r}_{t+\lambda})\right]}{4m} \leq \frac{1}{2}.
	\end{align}
	
	Let $\tau$ be the first time such that $\Phi(\mathbf{r}_{t+\tau}) \leq 4m$ occurs from any step $t$. Flip a coin after each of $\lambda$ steps until $\Phi(\mathbf{r}_{t+\tau}) \leq 4m$:
	\bgroup
	\arraycolsep=1.4pt
	\begin{equation}
	\Pr(\tau = i\lambda) \left\{
	\begin{array}{l l l}
	\geq p            \ \ \ & \mbox{if $i = 1$,}\\
	\leq p(1-p)^{i-1} \ \ \ & \mbox{if $i \geq 2$,}
	\end{array} \right.
	\end{equation}
	\egroup
	where $p=\Pr(\Phi(\mathbf{r}_{t+\lambda}) \leq 4m)$, thus $0.5 \leq p < 1$. Therefore,
	\begin{align}\nonumber
	\E[\tau] &= \sum_{i=1}^{\infty}i\lambda \Pr(\tau=i\lambda) \nonumber \\ 
	&\leq \lambda + \frac{1}{4}\lambda\sum_{i=2}^{\infty}(1-p)^{i-2} \nonumber \\
	&\leq \left(1+\frac{1}{4p}\right)\lambda \leq \frac{3}{2}\lambda.
	\end{align}
	
	By Markov's inequality, $\Pr(\tau \geq 60\lambda) \leq \frac{\E\left[\tau\right]}{60\lambda} = 0.025$. Conversely, $\Pr(\Phi(\mathbf{r}_{t+60\lambda}) \leq 4m ) \geq 0.975$. As a result, from any state, the system can rapidly (within $O(\log\log m)$ steps) converge to a $\epsilon$-Nash equilibrium, where the load ratios of the nodes are approximately balanced.
	
\end{proof}

\subsection{Incentive}
Note that non-cooperative or malicious behaviors in may significantly adversely affect the entire network via flooding the relay network with invalid messages. To this end, relayers will verify every message they received before propagating it to other relayers or nodes. This also enables relayers to identify malicious nodes, and then limit the effect of such attacks. Besides, It incurs bandwidth costs to deliver large traffic volumes to a large number of nodes. 

To ensure \fissions sustainability and robustness, relayers are incentivized to relay messages as quickly as possible, helping partition or block committees to win the race for block rewards (deterministically, the first $\theta$ votes for blocks). One possible approach is to apply a subscription-based model that every node will pay a rental fee to its relayer. Another approach is to open a payment channel between every node and its relayer such that each payment on a single message is paid in an on-demand basis. We will detail the incentive mechanism in our future work and implementations.

\section{The P2P Network}\label{sec:p2p}
Unlike other blockchains' P2P network, where messages are delivered in a gossip-like manner, \fissions P2P network takes advantage of P2P content-sharing capabilities of BitTorrent or IPFS. Combined with the relay network and the P2P network, \fission provides an efficient, cost-effective, and highly available distributed information storage and dissemination solution that scales to high transactions volumes.

Specifically, only hashes of messages are gossiped in \fissions P2P network. Hence \fission scales to large blocks. However, it may require every node to retrieve the whole content from either the relayer, or a node in the P2P network immediately after it receives a hash from its neighbor or relayer. For instance, a block committee member needs to retrieve the to-be-verified block and vote for it within a \textit{latency constraint} (\eg $\Delta_{\texttt{interim}}$ or $\Delta_{\texttt{main}}$) in order to receive the block rewards. To incentivize nodes to contribute their resources (\eg storage, bandwidth, \etc) to serve others, a small amount of fee (in \textsf{FIT}) will be charged by the \textit{content provider}, which are the nodes having the required data in their local storage. It should be noted that nodes should pay less fees to the content providers than the relayers (usually with much higher network capacities and hardware specifications) which provide a better QoS (\ie Quality-of-Service) in terms of latency.  

Therefore, the challenge is to minimize nodes' costs of retrieving content while respecting 1) the latency constraints of content, 2) the heterogeneous resource (\eg bandwidth) distribution on nodes, and 3) the system churn (nodes can join and leave the system anytime). To address this challenge, we propose an online algorithm that minimize the cost by optimally utilizing the bandwidth of nodes  in a fully decentralized manner. 

\subsection{Problem Formulation}
Let $V$ and $W$ denote the set of online nodes and the set of data (\eg transactions, blocks, \etc) that is required by these nodes, respectively. Let $P_k\in V$ denote the set of content providers who are online nodes that has data $k$ in their storage. The size of each data $i\in W$, denoted by $w_i$, is measured in units of kilobytes (\ie, KB). It should be noted that all the data is splitted into \textit{chunks}, in order to improve the data availability and efficiency of retrieval. Let $Q_i=\cup_{j\in V}\{q_i^j\}$ be the set of requests sent by all nodes for data $i$, where $q_i^j$ denotes the request sent by node $j$ for content $i$. The \emph{weight} of the request $q_i^j$ be the size of the data (\ie, $w_i$). The latency constraint for each data is denoted by $d=\min (\Delta_{\texttt{eager}}, \Delta_{\texttt{lazy}})$, measured in seconds. Without loss of generality, we assume online nodes are heterogeneous and have different upload capacities. Let $u_i$ denote the advertised upload capacity of node $i$, measured in kilobytes per second, \ie, KB/s. 

Since the amount of data that can be delivered by node $i$ within the latency constraint is $d\cdot u_i$, and thus the relayers will disseminate the rest of the amount of data, if any. Therefore, the relayers will send $\max\{l_i - d\cdot u_i, 0\}$ of data due to the inadequate upload capacity of node $i$. \textit{Our goal is to maximize the data sent by nodes while respecting the upload capacities of heterogeneous nodes and the latency constraints of data}, which is formulated as follows:
\begin{align}
\text{maximize:}   &\ \ \sum_{j\in V}\sum_{k\in W}w_kx_k^{j,i} \label{hdn:prob:max-p2p} \\
\mbox{s.t.:} &\ \ \sum_{i\in Q_k}\sum_{k\in W}w_kx_k^{j,i} \leq d\cdot u_i, \ \forall q_k^j\in Q_k, \\
&\ \ x_k^{j,i} \in \{0,1\}, \ \forall q_k^j\in Q_k.
\end{align}
where $x_k^{j,i}$ is an indicator variable indicating whether the request $q_k^j$ is sent to node $i$. Specifically, if $q_k^j$ is sent to node $i$, $x_k^{j,i}=1$; otherwise $x_k^{j,i}=0$.  It can be observed that the objective (\ref{hdn:prob:max-p2p}) is a multidimensional knapsack problem~\cite{puchinger2010multidimensional}, which is \textsf{NP}-complete~\cite{garey1979computers} in general. Many approaches have been proposed to solve it, such as LP-relaxation~\cite{schrijver1998theory}, the primal-dual method~\cite{buchbinder2009design} and dynamic programming~\cite{schrijver2003combinatorial}. However, these centralized approaches are not desirable and practical since a global connectivity information among nodes and the request information is required to obtain a feasible solution. Moreover, there are existing decentralized~\cite{panconesi2008fast} and online~\cite{buchbinder2009design,emek2010online} algorithms that compute approximate solutions within poly-logarithmic communication rounds, but they are based on the assumption that each node has the same capacity and is aware of global information, \ie, it knows about loads and upload capacities of nodes, and the connectivity information between any two of them. Note that the dual problem is to minimize the data sent by relayers, given the same request set, which can be formally stated as follows:
\begin{align}
\text{minimize:} &\ \ \sum_{i\in V} \max\{ \sum_{Q_k\neq \emptyset}\sum_{i\in Q_k}x_k^{j,i}w_k - d\cdot u_i, 0\}        \label{eq:hdn:p2p:obj}  \\ 
\mbox{s.t.:} &\ \ \sum_{i\in Q_k} x_k^{j,i} = 1, \ \forall r_k^j\in Q_k  \\
&\ \ x_k^{j,i} \in \{0,1\}, \ \forall q_k^j\in Q_k. 
\end{align}
Considering every request has only one provider. The load of node $i$ is $l_i=\sum_{Q_k\neq \emptyset}\sum_{i\in Q_k}x_k^{j,i}w_k$. Therefore, the objective (\ref{eq:hdn:p2p:obj})  can be simplified to
\begin{align}
\text{minimize:} \ \ \sum_{i\in V}\max\{l_i - d\cdot u_i, 0\},  \label{eq:hdn:p2p:obj2}
\end{align}
where $\max\{l_i - d\cdot u_i, 0\}$ is the amount of bandwidth consumed by the relayers due to the inadequate upload capacity of node $i$. Therefore, we transform the problem of maximizing the data sent by nodes into an asymmetric load balancing problem: how to distribute the loads on \emph{overloaded} nodes ($l_i > d\cdot u_i$) to \emph{underloaded} nodes ($l_i < d\cdot u_i$).

The objective of \fissions P2P storage layer is to design a P2P data retrieval strategy that can be implemented in a decentralized and online fashion. To this end, we model the P2P data retrieval problem as a congestion game and thus inherit its practical and decentralized nature, which has been successfully used to model load balancing problems in P2P networks as well as many other real world applications due to its conceptual simplicity.

\subsection{Data Retrieval Strategy (DRS)}
The objective (\ref{eq:hdn:p2p:obj2}) can be formulated as an asymmetric weighted congestion game, de=noted by $G$, which is defined as
\begin{align}\label{game}
G=(Q,P,(A_i)_{i\in Q},(f_e)_{e\in P}),
\end{align}
where $Q=\cup_{k\in W}Q_k$ and $P=\cup_{k\in P}P_k$ correspond to the set of requests and the set of content providers, respectively. Each request $i\in Q$ has a weight $w_i$ (\ie, the size of the requested data), and each node $i\in P$ has an upload capacity $u_i$. The strategy set of request $i$, denoted by $A_i$, is the set of nodes $P_i$ that have stored the requested data in their caches. Since every request is sent to only one node at a time, $G$ is a \emph{singleton} congestion game, \emph{i.e.}, $\forall i\in Q$, $a_i\in A_i$.

The strategy profile of $G$, denoted by $\mathbf{a} = (a_i)_{i\in Q}$, corresponds to the DRS\footnote{That is, choose which node to retrieve data from.}, and $\mathcal{A} = \times_{i\in Q}A_i$ corresponds to the set of all possible DRSs. The cost of request $i$ given a node selection profile $\mathbf{a}$ is defined as
\begin{align}\label{cost}
c_i(\mathbf{a}) = \min\{\max\{h_i - d\cdot u_j, 0\}, w_i\},
\end{align}
where $j=a_i$ is the node selected by request $i$ and $h_i$ is the \emph{height} \footnote{All nodes serve retrieval requests in a FIFO manner.} of request $i$ at node $j$. Clearly, for each node $j\in S$, the sum of the cost of requests that select node $j$ is
\begin{align}
\sum_{a_i=j}c_i(\mathbf{a}) &= \sum_{a_i=j} \min\{\max\{h_i - d\cdot u_j, 0\}, w_i\} \\
&=\max\{ l_j - d\cdot u_j, 0 \}, \ \ \forall a_i\in \mathbf{a}.
\end{align}
Recall that $\max\{l_j - d\cdot u_j, 0\}$ is the amount of data that the relayers will disseminate to the corresponding nodes due to the upload capacity limitation of node $j$. We define the potential function as the sum of the cost of all requests:
\begin{align}\label{potential}
\Phi(\mathbf{a}) &= \sum_{j\in P}\sum_{a_i=j}c_i(\mathbf{a}) = \sum_{j\in P}\max\{l_j - d\cdot u_j, 0\}.
\end{align}
Therefore, the problem of maximizing the data sent by nodes can be solved by minimizing the potential function of the corresponding congestion game. 

It has been shown that finding a pure Nash equilibrium in a congestion game is \textsf{PLS}-complete~\cite{fabrikant2004complexity}, hence the number of changes of the strategy profile required from one state to any pure Nash equilibrium is exponentially large. Therefore, our goal is to quickly reach an \emph{approximate} Nash equilibrium from any state by enabling each node to change its strategies for its requests according to a \emph{bounded jump rule} (see lines~\ref{alg:cond1} and~\ref{alg:cond2} in Algorithm \ref{alg:drs}), which will be discussed later.

It is worth to point out that the offline data retrieval strategies are undesirable in practice due to both the additional delay they incur to find appropriate content providers for all nodes, and the inefficient utilization of node bandwidth (since the node bandwidth cannot be used to delivery data during the process of finding a good node selection strategy). Furthermore, blockchains are highly dynamic environments such that churn and time-varying node capacities will significantly degrade the feasibility and efficiency of offline data retrieval strategies.

We introduce an online, light-weight data retrieval strategy which maximizes the data sent by nodes by enabling each node to repeatedly \emph{update} its strategy for each data that it requests. More specifically, all requests are sent by nodes in a repeated fashion, \ie, the strategy of each request can be changed if (1) the current content provider of the request is overloaded, and (2) the corresponding node finds an underloaded provider within the latency constraint.
\begin{algorithm}[ht!]
	\caption{Data Retrieval Strategy.} \label{alg:drs} Let $t$ be the local time of node $i\in N$\; Let $t_k$ be the time when $q_k^i$ is generated\; 
	Let $p_i^k$ be the provider of the request $q_k^i$\;
	\ForEach{$q_k^i\in Q_k$}{
		Contact a node $j$ from $P_k(d)$ u.a.r.\label{alg:fc} \;
		%\If{timeout}{Go to line \ref{alg:fc}\;}
		$p_i^k \leftarrow j$\;
		Send $q_k^i$ to node $p_i^k$\label{alg:drs:first}\;
		Let $h_k^i$ be height of $q_k^i$ at node $j$\;
		\If{\label{alg:nochange1}$h_k^i \leq d\cdot u_j$}{
			Stop contacting other nodes\;
		}
		\While{\label{alg:cond1}$t \leq t_k + d$ and  $h_k^i - (d - t + t_k)u_{p_i^k} \geq w_k $}{
			Contact a node $j^{\prime}$ from $P_k(d-t+t_k)$ u.a.r. \label{alg:rptreq} \;
			\If{timeout}{Go to line \ref{alg:cond1}\;}
			Let $l_{j^{\prime}}$ be the load of node $j^{\prime}$\;
			\If{\label{alg:cond2}$l_{j^{\prime}} \leq (d - t + t_k)u_{j^{\prime}}$}{
				$p_i^k \leftarrow j^{\prime}$\;
				Send $q_k^i$ to node $p_i^k$\;
			}\Else{Go to line \ref{alg:cond1}\;}
			Let $h_k^i$ be height of $q_k^i$ at node $j^{\prime}$\;
			\If{\label{alg:nochange2}$h_k^i < w_k + (d - t + t_k)\cdot u_j^{\prime}$}{
				Stop contacting other nodes\label{alg:stop2}\;
			}\label{alg:cond1-end}
		}
		Send the request for data $k$ directly to the relayer\;
	}
\end{algorithm}

Our algorithm is online such that each node performs data retrieval independently. Each request is sent to a content provider u.a.r. once it is generated (line \ref{alg:drs:first} in Algorithm~\ref{alg:drs}). Each node $i$ can change its content provider for each request in a repeated fashion if the bounded jump rule is satisfied. 

Our algorithm is light-weight such that each node $i$ needs to contact only one content provider in $P_k(d^{\prime})$ for each request $q_k^i$ with latency constraint $d^{\prime}$ at one time, where $P_k(d^{\prime}) = \{i\in P_k: l_i < d^{\prime}u_i \}$. The current latency constraint (\emph{i.e.}, $d^{\prime} = d - t + t_k$) of the request is calculated at the requesting node, where $t$ and $t_k$ are the current time and the time when the request is generated, respectively. It is worth noting that every node in $P_k$ has the same probability being selected, although some nodes have failed to satisfy the bounded jump rule in previous rounds. 

%The information of $S_k$ (the set of peers that have data $k$ in their caches) is maintained at the server and can be piggybacked onto the update packets in practical implementations. 

The key idea of the bounded jump rule is to restrict the change of content providers such that the potential of the system decreases monotonically. Since all requests are processed in a FIFO fashion, the nodes that have the requests which can be processed within their latency constraints will not consider changing their current strategies (see lines \ref{alg:nochange1} and \ref{alg:nochange2} in Algorithm \ref{alg:drs}). From the perspective of congestion games, those requesting nodes have no incentives to change their  strategies since their costs cannot be reduced (for pure Nash equilibrium) or not significantly reduced (for $\epsilon$-Nash equilibrium). More formally, for a given request $q_k^i$ (sent by node $i$ for data $k$) at node $p_i^k$ with a weight $w_k$ and a height $h_k^i(t)$, by Eq.~\ref{cost}, the cost of $q_k^i$ at time $t$ is
\begin{align}
c_{q_k^i}(\mathbf{a}(t)) = \min\{ \max\{ h_k^i(t) - (d-t+t_k)u_{p_i^k}, 0 \}, w_k \},
\end{align}
where $\mathbf{a}(t)$ is the strategy profile of all nodes at $t$, and $(d-t+t_k)$ is the latency constraint of data
$k$ at $t$. This implies
\begin{align}
c_{q_k^i}(\mathbf{a}(t)) = w_k \iff h_k^i(t) - (d-t+t_k)u_{p_i^k} \geq w_k. \label{alg:eq-bjr}
\end{align}

Therefore, a node will consider to change its strategy for its requests only when the costs of requests are higher than their weights, otherwise it will stop contacting other nodes (line \ref{alg:stop2} in Algorithm~\ref{alg:drs}). If the requester of $q_k^i$ (\emph{i.e.}, node $i$) finds another content provider $j^{\prime}$ at $t$ (line~\ref{alg:rptreq} in Algorithm~\ref{alg:drs}) during the repeated update process (the loop from line~\ref{alg:cond1} to line \ref{alg:stop2} in Algorithm~\ref{alg:drs}) such that $l_{j^{\prime}}(t) \leq (d - t + t_k)u_{j^{\prime}}$, it will send $q_k^i$ to node $j^{\prime}$. It is worth noting that multiple nodes may find the content provider $j^{\prime}$ and send requests to it at the same time. In the worst case, only the first request can be processed by node $j^{\prime}$ since nodes process requests in a FIFO manner. Finally, the request will be sent to the relayers if $c_{q_k^i}(\mathbf{a}(t))= w_k$ for any $t\in[t_k, t_k+d]$. It should be noted that no requesting node will change its content provider mid-way during a download, since it will stop contacting other content providers if it can download the request data from the current content provider (lines~\ref{alg:nochange1} and~\ref{alg:nochange2} in Algorithm~\ref{alg:drs}).

With the bounded jump rule, it is easy to observe that the potential function of game $G$ is monotonically decreasing such that for a given set of requests and $\forall t_1,t_2$, we have \begin{align}
t_1 \leq t_2 \implies \Phi(\mathbf{a}(t_1)) \geq \Phi(\mathbf{a}(t_2)).
\end{align}
Furthermore, our algorithm applies to a fully distributed and concurrent setting such that all nodes can change their content providers at the same time without a centralized coordination.

\subsection{Analysis}
%We formulate the problem of maximizing the data sent by nodes as a congestion game, and we define the potential of the congestion game as the data sent by relayers (see Eq. \ref{potential}), which is minimized when a Nash equilibrium (a steady state) is reached. 
With the proposed bounded jump rule, a steady state is a $\epsilon$-Nash equilibrium of the congestion game. Therefore, a critical question to investigate is \emph{how fast does the system reach a $\epsilon$-Nash equilibrium from any state?} In what follows, we seek to investigate the upper bound of the convergence time of our data retrieval strategy.

Recall that the potential function of the game $G$ with the bounded jump rule is monotonically decreasing in the presence of concurrent changes of nodes' strategies (nodes update their content providers individually), because nodes only consider to change the overloaded content providers. More specifically, node $i$ will change the content provider of request $q_k^i$,  if and only if $h_k^{i,j} - (d-t+t_k)u_j \geq w_k$ (\ie Eq.~\ref{alg:eq-bjr} and Line~\ref{alg:cond1} in Algorithm~\ref{alg:drs}), where $t$ and $t_k$ are the local time and the time when the request is generated, respectively. The equivalent condition is $h_k^{i,j} + (t-t_k)u_j - w_k \geq d\cdot u_j$. Therefore, we can analyze the convergence rate of the proposed node selection strategy via analyzing the convergence rate of the corresponding congestion game.

\begin{theorem}\label{thm:convergence}
	Given a game $G=(Q,P,(A_i)_{i\in Q},(f_e)_{e\in P})$ that satisfies the bounded jump rule, where $Q=\cup_{k\in W}Q_k$ and $P=\cup_{k\in W}P_k$ is the set of requests and content providers, respectively. Then a $\epsilon$-Nash equilibrium can be reached within $O(\log n)$ rounds, where $n=|V|$ is the number of online nodes.
\end{theorem}
\begin{proof}
	
	Let $\Phi(t)$ be the potential of the game at time $t$. According to Eq.~\ref{potential}, we have $\Phi(t)=\sum_{k\in Q(t)}w_k$. Similarly, let $P(t)=\cup_{k\in Q(t)}P_k(d)$ be the set of underloaded nodes at time $t$, where $P_k(d)=\{i\in S_k: l_i < d\cdot u_i\}$ denotes the set of underloaded nodes that have data $k$ in their caches. Let $n=|V|$ and $m_t=|P(t)|$ be the number of online nodes and underloaded content providers at time $t$, respectively. Note that a pure equilibrium is reached at time $t$ if (1) $\Phi(t)=0$ (all requests can be delivered by nodes within latency constraints), or (2) $m_t=0$ (no more underloaded content providers). Therefore, we introduce a non-increasing function $\Omega(t) = m_t\Phi(t)$ to analyze the convergence rate of DRS.
	
	At time $t+1$, let $\mathbf{x}=\{x_1,x_2,...,x_{m_t}\}$ be a vector, where $x_i$ denote the total size of requests that migrate to content provider $i\in P(t)$. Obviously, $\Phi(t)=\sum_{i=1}^{m_t}x_i$. By applying Cauchy-Schwarz inequality to $\E[\Omega(t+1)]$, which is the expected value of $\Omega(t+1)$, we have the following:
	\begin{align}
	\E[\Omega(t+1)] &= \E[m_{t+1}\Phi(t+1)] \nonumber\\
	&\leq \left(\E[m_{t+1}^2]\E[\Phi^2(t+1)]\right)^{\frac{1}{2}}. \label{eqn:psi-1}
	\end{align}
	Noting that the square root function $f:x\rightarrow \sqrt{x}$ and the function $g:x\rightarrow x^2$ are convex functions, we apply Jensen's inequality to Eq.~\ref{eqn:psi-1} and have the following:
	\begin{align}
	&\E[\Omega(t+1)] \leq \E[m_{t+1}]\E[\Phi(t+1)] \nonumber \\
	&=\sum_{m=1}^{m_t}m\Pr(m_{t+1}=m)\E[\Phi(t+1)|m_{t+1}=m], \label{eqn:psi-2}
	\end{align}
	where $m$ is the number of underloaded content providers in $P(t+1)$ at time $t+1$. The expected value of the potential function given $m$ is
	\begin{align}
	\E[\Phi(t+1)|m_{t+1} = m] &= \Phi(t) - \frac{\dbinom{m-1}{m_t-1}}{\dbinom{m}{m_t}}\sum_{i=1}^{m_t}x_i \nonumber \\
	&= \left( 1 - \frac{m}{m_t}\right)\Phi(t). \label{eqn:psi-3}
	\end{align}
	Combining Eq.~\ref{eqn:psi-3} and Eq.~\ref{eqn:psi-2}, we have the following:
	\begin{align}
	\E[\Omega(t+1)] &\leq \sum_{m=1}^{m_t}m\left( 1 - \frac{m}{m_t}\right)\Phi(t)\Pr(m_{t+1} = m) \\
	&\leq \frac{m_t}{4}\Phi(t)\sum_{m=1}^{m_t}\Pr(m_{t+1} = m) \leq \frac{1}{4}\Omega(t). \label{eqn:psi-4}
	\end{align}
	
	Let $\tau$ be the number of rounds that a $\epsilon$-Nash equilibrium (\emph{i.e.}, $\Omega(t+\tau)=O(1)$) is reached from time $t$. By Lemma~\ref{lem}, we have the following
	\begin{align}
	\E[\tau| \Omega(t)] \leq \log\left( \frac{\Omega(t)}{\Omega(t+\tau)} \right) = \log\left( \Psi(t) \right).
	\end{align}
	Note that the maximum value of of $\Omega(t)$ is $|W|n^2$ ($|W|$ is the total number of data), which occurs in the case when every node requests all the data in the region. Therefore, $\E[\tau]\leq 2\log(|W|n)$. By Markov's inequality, we have that
	\begin{align}
	\Pr(\tau\geq 40\log(|W|n))\leq \frac{\E[\tau]}{40\log(|W|n)} =0.05.
	\end{align}
	This implies that $\Pr(\tau\leq 40\log(|W|n)) \geq 0.95$. As a result, with a high probability, a $\epsilon$-Nash equilibrium can be reached from any state within $O(\log n)$ rounds.
\end{proof}

\begin{lemma}\label{lem}
	Let $X_1,X_2,...$ denote a sequence of non-negative random variables and assume that for all $i \geq 0$
	\begin{align}
	\E[X_i|X_{i-1} = x_{i-1}] \leq \alpha\cdot x_{i-1}
	\end{align}
	for some constant $\alpha\in (0,1)$. Furthermore, fix some constant $x^{*}\in (0,x_0]$ and let $\tau$ be the random variable that describes the smallest $t$ such that $X_t \leq x^{*}$. Then,
	\begin{align}
	\E[\tau|X_0=x_0] \leq \frac{2}{\log(1/\alpha)}\cdot \log\left(\frac{x_0}{x^{*}}\right).
	\end{align}
\end{lemma}
\begin{proof}
	The complete proof is shown in~\cite{fischer2008approximating}.
\end{proof}

Theorem \ref{thm:convergence} indicates that the proposed DRS can rapidly reach a $\epsilon$-Nash equilibrium within $O(\log n)$ communication rounds from any state, where $n=|V|$ is the total number of online nodes in the system.

\section{Future Work}
At the time of writing, \fission is still a Proof of Concept, and it has limitations that we want to address in future work. As described in Section~\ref{sec:p2p}, nodes are contributing their resources (\eg storage and bandwidth) to speed up information propagation and improve data availability. We leave to future work the exploration of incentive mechanisms that encourages nodes to share their resources, similar to Filecoin\footnote{https://filecoin.io/}. We also leave to future work the use of advanced cryptography, such as BLS signature scheme~\cite{boneh2001short} for performance improvements on \ppap consensus. Last but not least, we see in Section~\ref{sec:ppap:security} that the system activity (\ie online users participating the consensus process) has a great impact to security guarantees. Considering that most of users may not be online 24x7, we will investigate some delegation mechanisms that improve the security of our consensus protocol without significantly degrading the decentralization nature of \fission.

\section{Conclusion}\label{sec:conclusion}
In this paper, we present \fission that distinguishes itself from all existing permissionless blockchains that it achieves scalability in both terms of system throughput and transaction confirmation time through maximizing the parallelization (Eager-Lazy pipeling) and efficiency (adaptive partitioning and PoS-based consensus protocol) in the computation layer, and minimizing the information propagation latency via a hybrid topology.

\bibliographystyle{abbrv}
{\Large \bibliography{fission} } 
\appendix

\section{Lemmas for Proof of Theorem~\ref{theorem:plb:convergence}}\label{appendix:hdn-relay}

Let $J:=\{j \mid r_j < r_i \}$ be the set of relay nodes that have a smaller load ratio than $m_i$. Similarly let $H:=\{h \mid r_h > r_i \}$ be the set of relay nodes that have a larger load ratio than $m_i$. Let $m:= |M|$ be the total number of relay nodes.

\begin{lemma}\label{lemma:plb:exp}
	Let $\mathbf{r}_t$ be the load ratio vector at time $t$. Then, the expected load ratio of any relay node $m_i$ at time $t+1$, denoted by $\mathbf{r}_{t+1}^i$, is the optimal load vector $\overline{r}$, i.e., $\forall r_i\in \mathbf{r}, \ \E\left[ \mathbf{r}_{t+1}^i \mid \mathbf{r}_t\right] = \overline{r}$.
\end{lemma}
\begin{proof}
	For any relay node $m_k\in M$, the nodes that choose $m_k$ as a helper at time $t$ will choose $m_i$ as the relay node with probability $\Pr(k,j)$ (see Eq.~\ref{eq:relay:relay-selection}), hence we have
	\begin{align}
	\E[\mathbf{r}_{t+1}^i \mid \mathbf{r}_t] &= \frac{1}{u_i}\left(\sum_{k=1}^ml_k\Pr(k,i)\right) \nonumber\\
	&= \frac{1}{u_i} \left( \sum_{h\in H}l_h\Pr(h,i) + l_i - \sum_{j\in J}l_i\Pr(i,j) \right)\ \nonumber\\
	&= \frac{1}{u_i} \left( \frac{u_i}{u}\sum_{k\in H\cup J}l_k + l_i - \frac{l_i}{u}\sum_{k\in H\cup J}u_k \right) \nonumber\\
	&= \frac{1}{u_i} \left( \frac{u_i}{u}l - \frac{u_i}{u}\sum_{k\not\in H\cup J}l_k + \frac{l_i}{u}\sum_{k\not\in H\cup J}u_k\right).
	\end{align}
	For each node $k\not\in H\cup J$, $r_k=r_i$. Hence $\frac{u_i}{u}\sum_{k\not\in H\cup J}l_k -
	\frac{l_i}{u}\sum_{k\not\in H\cup J}u_k = 0$. As a result, $\E\left[ \mathbf{r}_{t+1}^i \mid
	\mathbf{r}_t\right] = \frac{l}{u} = \overline{r}$.
\end{proof}

\begin{lemma}\label{lemma:plb:var}
	$\sum_{i=1}^m\Var\left[ \mathbf{r}_{t+1}^i \mid \mathbf{r}_t \right] \leq \left(m\Phi(\mathbf{r}_t)\right)^\frac{1}{2}$. 
\end{lemma}
\begin{proof}
	Let $(Y_{i,1},\dots,Y_{i,m})$ be a random variable drawn from a multinomial distribution with the constraint $\sum_{j=1}^{m}Y_{i,j}=l_i$. For each relay node $m_i$, the variance of the load ratio of $m_i$ at time $t+1$ is shown as below:
	\begin{align}
	\Var[\mathbf{r}_{t+1}^i \mid \mathbf{r}_t] 
	&= \sum_{k=1}^m\Var(\frac{1}{u_i}Y_{k,i}) = \frac{1}{u_i^2}\sum_{k=1}^m\Var(Y_{k,i}) \nonumber \\
	&= \frac{1}{u_i^2} \left( \sum_{k=1}^ml_k\Pr(k,i)(1-\Pr(k,i)) \right)  \nonumber  \\
	&\leq \frac{1}{u_i^2} \left( \sum_{h\in H}l_h\Pr(h,i) + l_i(1-\Pr(i,i)) \right) \nonumber \\
	%&= \frac{1}{u_i^2} \left( \sum_{l\in L}x_l\Pr(l,i) + \sum_{j\in J}l_i\Pr(i,j) \right) \nonumber \\
	&= \frac{1}{u_i^2} \left( \sum_{h\in H}l_h\frac{u_i}{u}(1-\frac{r_i}{r_h})
	+ \sum_{j\in J}l_i\frac{u_j}{u}(1-\frac{r_j}{r_i}) \right) \nonumber \\
	&= \frac{1}{u_iu} \left( \sum_{h\in H}\left(l_h - r_iu_h\right) + \sum_{j\in J}\left(r_iu_j - l_j\right) \right) \nonumber \\
	&= \frac{1}{u_iu} \sum_{k=1}^m|l_k - r_iu_k| = \frac{1}{u} \sum_{k=1}^m\frac{u_k}{u_i}|r_k - r_i|, \nonumber
	\end{align}
	where $u_i$ is the upload capacity of relay node $m_i$, and we assume $u_i\geq 2$. Therefore, we sum up the variance of the load ratio at each relay node and we have
	\begin{align}
	%\frac{1}{u} \sum_{k=1}^m\frac{u_k}{u_i}|r_k - r_i|
	\sum_{i=1}^m\Var[\mathbf{r}_{t+1}^i \mid \mathbf{r}_t]
	&\leq \frac{1}{u} \sum_{i=1}^m\sum_{k=1}^m\frac{u_k}{u_i}|r_k - r_i|. \label{inq:var}
	\end{align}
	
	For all $r_k,r_i\in \mathbf{r}$, $|r_k - r_i| \leq |r_k - \overline{r}| + |r_i - \overline{r}|$. Using this, we can simplify the above inequality:
	\begin{align}
	%\vspace{-10pt} \frac{1}{u} \sum_{k=1}^m\frac{u_k}{u_i}|r_k - r_i|
	\sum_{i=1}^m\Var[\mathbf{r}_{t+1}^i \mid \mathbf{r}_t]
	%&\leq \frac{1}{u} \sum_{i=1}^m\sum_{k=1}^m\frac{u_k}{u_i}|r_k - r_i| \nonumber\\
	%\vspace{-10pt} \frac{1}{u} \sum_{i=1}^m\sum_{k=1}^m\frac{u_k}{u_i}|r_k - r_i|
	&\leq \frac{1}{u} \sum_{i=1}^m\sum_{k=1}^m\frac{u_k}{u_i}\left( |r_k - \overline{r}| + |r_i - \overline{r}| \right) \nonumber \\
	&= \sum_{i=1}^m\frac{1}{u_i}|r_i - \overline{r}| + \frac{1}{u} \sum_{k=1}^m\sum_{i=1}^m\frac{u_k}{u_i} |r_k - \overline{r}|. \nonumber
	\end{align}
	
	Note that the multivariable function $f(u_1,...,u_m)=\sum_{i=1}^m\frac{1}{u_i}$ is a Schur-concave\cite{marshall1979inequalities} function, and the maximum of $f$, denoted by $\sup f$, is achieved if and only if
	$\forall i,j\in[1,m], \ u_i=u_j$ (by Karamata's inequality\cite{marshall1979inequalities}). That is $\sup f =
	f(\frac{u}{m},...,\frac{u}{m})$. Hence we have
	\begin{align}
	\frac{1}{u} \sum_{k=1}^m\sum_{i=1}^m\frac{u_k}{u_i} |r_k - \overline{r}|
	&\leq \frac{m}{u^2}\sum_{k=1}^m u_k|r_k - \overline{r}|.
	\end{align}
	
	Since $u_i \geq 2$ for each $m_i\in M$, we have $\sum_{i=1}^m\frac{1}{u_i}|r_i - \overline{r}| \leq \frac{1}{2}|r_i - \overline{r}|$. In addition, $2mu_k\leq u^2$ for any $u_k\geq 2$ and $m\geq 2$. Combining all the inequations above, we have
	\begin{align}
	\sum_{i=1}^m\Var[\mathbf{r}_{t+1}^i \mid \mathbf{r}_t]
	&\leq \sum_{i=1}^m|r_i - \overline{r}|
	\end{align}
	
	Then, by Cauchy-Schwarz inequality, we obtain
	\begin{align}
	\sum_{i=1}^m\Var[\mathbf{r}_{t+1}^i \mid \mathbf{r}_t]
	&\leq \left( m\sum_{i=1}^m | \overline{r} - r_i |^2 \right)^{\frac{1}{2}} \leq \left( m\Phi(\mathbf{r}_t) \right)^{\frac{1}{2}}
	\end{align}
	
\end{proof}

\end{document}